\documentclass[preprint]{elsarticle}

\usepackage[T1]{fontenc}
\usepackage[utf8]{inputenc}
\usepackage[english]{babel}

\usepackage{times}
\usepackage{amsmath}
\usepackage{amssymb}
\usepackage{amscd}
\usepackage{latexsym}
\usepackage{tabularx}
\usepackage{enumitem}
\usepackage{stmaryrd}
\usepackage[pdftex]{thumbpdf}

\usepackage{color}
\usepackage{vmargin}
\usepackage{epsfig}
\usepackage{pgf}
\usepackage{pgfcore}
\usepackage{pgfbaseimage}
\usepackage{pgfbaselayers}
\usepackage{pgfbasepatterns}
\usepackage{pgfbaseplot}
\usepackage{pgfbaseshapes}
\usepackage{pgfbasesnakes}
\usepackage{tikz}
\usepackage{natbib}
\usetikzlibrary{snakes}

\newenvironment{proof}{\par \noindent \textbf{Proof} \\}{\hfill$\Box$}
\newcommand{\lb}{\llbracket}
\newcommand{\rb}{\rrbracket}

\newtheorem{corollary}{Corollary}
\newtheorem{definition}{Definition}

\newtheorem{proposition}{Proposition}
\newtheorem{theorem}{Theorem}
\newtheorem{lemma}{Lemma}

\newcommand{\cy}{\mathcal{C}}

\newcommand{\wc}[1]{W\hspace{-0.6mm}C_{#1}}
\newcommand{\ic}[1]{I\hspace{-0.3mm}C_{#1}}
\newcommand{\st}{\:|\:}
\newcommand{\ie}{\emph{i.e.}}
\newcommand{\eg}{\emph{e.g.}}

\journal{Discrete Mathematics}
\begin{document}
\begin{frontmatter}

\title{On two variations of identifying codes\tnoteref{t1}}

 \tnotetext[t1]{This research is
    supported by the ANR Project IDEA {\scriptsize $\bullet$} Identifying coDes in Evolving
    grAphs {\scriptsize $\bullet$} {ANR-08-EMER-007},  2009-2011.}

\author[bdx]{Olivier Delmas}
\author[gre]{Sylvain Gravier}
\author[bdx]{Mickael Montassier}
\author[gre]{Aline Parreau}
\ead{aline.parreau@ujf-grenoble.fr}

\address[bdx]{Universit\'e de Bordeaux, LaBRI, 351 
cours de la Lib\'eration, 33400 {Talence}, France}

\address[gre]{Institut Fourier (\textsc{umr} {\scriptsize 5582}), 100 rue des Maths, BP 74,
38402 {Saint-Martin d'H\`eres}, France}

\begin{abstract}
    Identifying codes have been introduced in 1998 to model
    fault-detection in multiprocessor systems. In this paper, we introduce
    two variations of identifying codes: weak
codes and light codes. They correspond to fault-detection by
successive rounds. We give exact bounds for those two definitions for the 
family of cycles.
\end{abstract}


\end{frontmatter}

\section{Introduction}
Identifying codes are dominating sets having the property that any two
vertices of the graph have distinct neighborhoods within the
identifying code. Also, they can be used to uniquely identify or
locate the vertices of a graph. Identifying codes have been introduced
in 1998 in \cite{KCL98} to model fault-detection in multiprocessor
systems.  Numerous papers already deal with identifying codes (see
\eg{} \cite{bibLobstein} for an up-to-date bibliography).  A
multiprocessor system can be modeled as a graph where vertices are
processors and edges are links between processors.  Assume now that at
most one of the processors is defective, we would like to locate it by
testing the system.  For this purpose, we select some processors
(constituting the code) and have them test their $r$-neighborhoods
(\ie{} the processors at distance at most $r$). The processor sends an
alarm if it detects a fault in its neighborhood.  We require that we
can, with these answers, tell if there is a faulty processor and, in
this case, locate it uniquely.  This corresponds exactly to finding an
identifying $r$-code of the graph of the system.

\ 

Assume now that a processor can restrict its tests to its
$i$-neighborhood for $i\in \lb 0,r \rb$. Then, we can have a detection
process by rounds: at the first step, the selected processors test
their $0$-neighborhoods, then they test their $1$-neighborhoods,
\dots, until the $r$-neighborhoods.  We stop the process when we can
locate the faulty processor.  We introduce in this paper \emph{weak
  $r$-codes} (\emph{resp.} \emph{light $r$-codes}) that will model this
process without memory, \ie{} to identify a faulty processor at the round
$i$, the supervisor does not need to remember the collected
information of the rounds $j < i$ (\emph{resp.} with memory, \ie{} to identify a faulty processor at the round
$i$, the supervisor needs to remember the collected
information of the rounds $j < i$) and study them
for the family of cycles.

\ 

Let us give some notations and definitions. We denote by $G=(V,E)$ a simple non oriented graph
having vertex set $V$ and edge set $E$. Let $x$ and $y$ be two vertices of
$G$. The \emph{distance} 
$d(x,y)$  between $x$ and $y$ is
the number of edges of a shortest path between $x$ and $y$. Let $r$ be
an integer. The \emph{ball} centered on $x$ of radius $r$, denoted by $B_r(x)$ is
defined by $B_r(x)=\{y \in V \st d(x,y)\leq r\}$.

\ 

An \emph{$r$-dominating set} of $G$ is a subset $C\subseteq V$ such that
$\cup_{c\in C}B_r(c) =V$. This means that each vertex of $G$ is at distance at most $r$ of a vertex of $C$.
We say that a subset $C\subseteq V$ \emph{$r$-separates}
$x$ and $y$ 
if and only if $B_r(x)\cap C \neq B_r(y) \cap C$ (we will also say in 
this case that ``$x$ and $y$ are separated by $C$ for radius $r$'' or
that ``$x$ is separated from $y$ by $C$ for radius $r$''). 
A set $C$ \emph{$r$-identifies} $x$ if and only if it $r$-separates $x$ 
from all the other vertices. 

\paragraph{\bf (1) Identifying $r$-code}
An \emph{identifying $r$-code} of $G$ is an $r$-dominating set $C\subseteq V$ that
$r$-identifies all the vertices:

$$
\forall x \in V, \forall y \neq x \in V, B_{r}(x)\cap C \neq B_{r}(y)\cap C 
$$

\paragraph{\bf (2) Weak $r$-code}
A \emph{weak $r$-code} of $G$ is a $r$-dominating set $C \subseteq V$ such that 
each vertex $x$ is $r_x$-identified by $C$ for some radius $r_x \in \lb 0,r \rb$:
$$
\forall x \in V, \exists r_x \in \lb 0,r \rb, s.t. \  \forall y\neq x \in V, B_{r_x}(x)\cap C \neq B_{r_x}(y)\cap C
$$

\paragraph{\bf (3) Light $r$-code}
A \emph{light $r$-code} of $G$ is a $r$-dominating set $C\subseteq V$ such that
each
 pair $(x,y)$ of vertices is $r_{xy}$-separated by $C$ for some radius $r_{xy} \in \lb 0,r \rb$:

$$
\forall x \in V, \forall y\neq x \in V, \exists r_{xy} \in \lb 0,r \rb, s.t. \ B_{r_{xy}}(x)\cap C \neq B_{r_{xy}}(y)\cap C 
$$

\bigskip

Figure \ref{fig:ex1} gives an example of a weak $2$-code of $P_5$ (elements of the code are in black, 
as in all the figures).
Indeed, vertices $v_3$ and $v_4$ are identified  for radius $0$, vertices $v_2$ and $v_5$ 
are identified for radius $1$ and vertex  $v_1$ is identified for radius $2$. 
But this code is not an identifying $2$-code of $P_5$: vertices $v_2,v_3,v_4$ and $v_5$ are not separated for radius $2$.
Figure \ref{fig:ex2} gives a light $2$-code of $P_5$ which is not a weak $2$-code: vertex $v_2$ is separated 
from vertex 
$v_1$ only for radius $0$ and for this radius, vertex $v_2$ is not separated 
from $v_3$.

\bigskip

\begin{figure}[h]
  \begin{center}
    \begin{minipage}[c]{.46\linewidth}
      \centering
      \begin{tikzpicture}[scale=1]
        \path[draw] (0,0) -- (4,0);
        \tikzstyle{every node}=[shape=circle,fill=white,draw=black,minimum
        size=0.5pt,inner
        sep=1.5pt]
        \node at (0,0) {};
        \node at (1,0) {};
        \node at (2,0) {};
        \node at (3,0) {};
        \node at (4,0) {};
        
        \tikzstyle{every node}=[]
        \node at (0,-0.5) {$v_1$};
        \node at (1,-0.5) {$v_2$};
        \node at (2,-0.5) {$v_3$};
        \node at (3,-0.5) {$v_4$};
        \node at (4,-0.5) {$v_5$};
       
        \tikzstyle{every node}=[shape=circle,fill=black,draw=black,minimum
        size=0.5pt,inner
        sep=1.5pt]
        \node at (3,0) {};
        \node at (2,0) {};
    \end{tikzpicture}
    \caption{\label{fig:ex1}A weak 2-code that is not an identifying $2$-code }
  \end{minipage} \hfill
  \begin{minipage}[c]{.46\linewidth}
    \centering
    \begin{tikzpicture}[scale=1]
      \path[draw] (0,0) -- (4,0);
        \tikzstyle{every node}=[shape=circle,fill=white,draw=black,minimum
        size=0.5pt,inner
        sep=1.5pt]
        \node at (0,0) {};
        \node at (1,0) {};
        \node at (2,0) {};
        \node at (3,0) {};
        \node at (4,0) {};
     
        \tikzstyle{every node}=[]
        \node at (0,-0.5) {$v_1$};
        \node at (1,-0.5) {$v_2$};
        \node at (2,-0.5) {$v_3$};
        \node at (3,-0.5) {$v_4$};
        \node at (4,-0.5) {$v_5$};
       
        \tikzstyle{every node}=[shape=circle,fill=black,draw=black,minimum
        size=0.5pt,inner
        sep=1.5pt]
        \node at (0,0) {};
        \node at (4,0) {};

    \end{tikzpicture}
    \caption{\label{fig:ex2}A light $2$-code that is not a weak $2$-code }
   \end{minipage}
\end{center}
\end{figure}

\bigskip
A code $C$ is said to be \emph{optimum} if its cardinality is minimum. We denote by $\ic{r}(G)$ 
(\emph{resp.} $\wc{r}(G)$, $LC_r(G)$) the cardinality of an optimum identifying (\emph{resp.} weak, light) $r$-code.
An identifying $r$-code is a weak $r$-code and a weak $r$-code is a light $r$-code. 
This implies the following inequality: $\ic{r}(G)\geq \wc{r}(G) \geq LC_r(G)$. 
For all graphs and for any $r$, there exits a weak $r$-code 
and a light $r$-code (using for instance all the vertices as the code), whereas this is not true for identifying codes.

\bigskip

Let us now give some bounds for weak codes.

\begin{theorem}
 Let $r$ and $k$ be two integers and $w_r(k)$ be the maximum order of a graph $G$ such that $G$ has a weak $r$-code of size $k$. We have: $$w_r(k)=k+r(2^k-2)$$
\end{theorem}

\begin{proof}
First, we construct a graph $H_r^k$ in the following way (see Figure~\ref{fig:borne} for $r=4$ and $k=3$).
The graph $H_r^k$ has vertex set $C\cup I_1\cup ... \cup I_r$  where 
$C=\{1,...,k\}$ and $I_{j}$ has size $2^k-2$ for $1\leq j \leq r$. 
Each vertex of $I_j$ corresponds to a non-empty strict subset of $\{1,...,k\}$.
Each vertex of $I_1$ is linked to the vertices of $C$ that form its subset, and
each vertex of $I_{j}$ for $j>1$ is linked to the vertex of $I_{j-1}$ that corresponds to the
same subset.
Furthermore, $C$ induce a clique in $H_r^k$. 
The graph $H_r^k$ has order $k+r(2^k-2)$ and one can check that $C$ is a weak $r$-code of $H_r^k$ (a vertex of $I_j$ is identified for radius $j$). So $w_r(k)\geq k+r(2^k-2)$. 

\ 


Now let $G$ be a graph and $C$ be a weak $r$-code size $k$ of $G$. Let us try to maximize the number of identified vertices for each radius $i \leq r$.
\begin{itemize}[noitemsep]
\item For radius $0$, only the $k$ vertices of $C$ can be identified.
\item For radius $1$, at most $2^k$ additionnal vertices can be identified (one for each subset of $C$). 
However, it is not possible to have all the subsets. Indeed,
all the elements of $\{B_1(c)\cap C \st c \in C\}$ cannot be used to identify a vertex not in $C$ for radius $1$.

If $2^k-1$ additional vertices are identified at radius $1$, that
means that $\{B_1(c)\cap C \st c \in C\}$ contains only one element,
which is necessarily the whole set $C$.  Then all the strict subsets
of $\{1,...,k\}$ are used to identify a vertex for radius $1$, in
particular, one vertex is identified by the emptyset and so is not
$1$-dominated by $C$.  As the set $C$ must be an $r$-dominating set,
then $r\geq 2$.  Furthermore, if we try to add a new vertex $x$ in
$G$, then necessarily, $B_1(x)\cap C=C$ and $x$ will not be identified
for any radius. So, $G$ has order $k+(2^k-1)$ and $r \geq 2$. A
contradiction with the bound $w_r(k)\geq k+2(2^k-2)$ for $r\geq 2$,
given by the construction of the graph $H_2^k$. 
It follows that at most $2^k-2$ additionnal vertices are identified
for radius $1$ in $G$.

\item For radius $2 \leq i \leq r$, using a similar process, we can show that at
most $2^k-2$ vertices are identified at round $i$.
\end{itemize}

Summing the number of identified vertices at each round, we obtain that 
$G$ has order at most $k+r(2^k-2)$.  It follows that $w_r(k)=
k+r(2^k-2)$.

\end{proof}

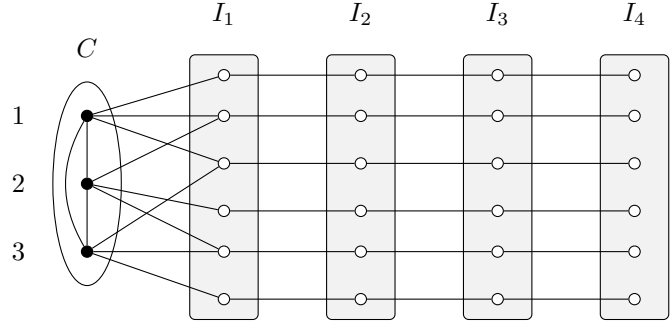
\begin{figure}[h]
\centering
    \begin{tikzpicture}[scale=0.9]

\fill[black!5,rounded corners=2pt] (1.5,0) rectangle (2.5,3.9);
\draw[rounded corners=2pt] (1.5,0) rectangle (2.5,3.9);

\fill[black!5,rounded corners=2pt] (3.5,0) rectangle (4.5,3.9);
\draw[rounded corners=2pt] (3.5,0) rectangle (4.5,3.9);

\fill[black!5,rounded corners=2pt] (5.5,0) rectangle (6.5,3.9);
\draw[rounded corners=2pt] (5.5,0) rectangle (6.5,3.9);

\fill[black!5,rounded corners=2pt] (7.5,0) rectangle (8.5,3.9);
\draw[rounded corners=2pt] (7.5,0) rectangle (8.5,3.9);

\node[] at (0,4) {$C$};
\node[] at (2,4.5) {$I_1$};
\node[] at (4,4.5) {$I_2$};
\node[] at (6,4.5) {$I_3$};
\node[] at (8,4.5) {$I_4$};
\node[] at (-1,3) {$1$};
\node[] at (-1,2) {$2$};
\node[] at (-1,1) {$3$};

\tikzstyle{every node}=[shape=circle,fill=black,draw=black,minimum
size=2pt,inner sep=1.5pt]
\node(1) at (0,3) {};
\node(2) at (0,2) {};
\node(3) at (0,1) {};

\draw (0,2) ellipse(0.5cm and 1.5cm);

\tikzstyle{every node}=[shape=circle,fill=white,draw=black,minimum
size=2pt,inner sep=1.5pt]
\node(100) at (2,3.6) {};
\node(010) at (2,1.6) {};
\node(001) at (2,0.3) {};
\node(110) at (2,3) {};
\node(011) at (2,1) {};
\node(101) at (2,2.3) {};

\node(1001) at (4,3.6) {};
\node(0101) at (4,1.6) {};
\node(0011) at (4,0.3) {};
\node(1101) at (4,3) {};
\node(0111) at (4,1) {};
\node(1011) at (4,2.3) {};

\node(1002) at (6,3.6) {};
\node(0102) at (6,1.6) {};
\node(0012) at (6,0.3) {};
\node(1102) at (6,3) {};
\node(0112) at (6,1) {};
\node(1012) at (6,2.3) {};

\node(1003) at (8,3.6) {};
\node(0103) at (8,1.6) {};
\node(0013) at (8,0.3) {};
\node(1103) at (8,3) {};
\node(0113) at (8,1) {};
\node(1013) at (8,2.3) {};

\path[draw] (1) to (100);
\path[draw] (1) to (110);
\path[draw] (1) to (101);
\path[draw] (2) to (010);
\path[draw] (2) to (110);
\path[draw] (2) to (011);
\path[draw] (3) to (001);
\path[draw] (3) to (011);
\path[draw] (3) to (101);

\path[draw] (1001) to (100);
\path[draw] (1101) to (110);
\path[draw] (1011) to (101);
\path[draw] (0101) to (010);
\path[draw] (0111) to (011);
\path[draw] (0011) to (001);
\path[draw] (1001) to (1002);
\path[draw] (1101) to (1102);
\path[draw] (1011) to (1012);
\path[draw] (0101) to (0102);
\path[draw] (0111) to (0112);
\path[draw] (0011) to (0012);

\path[draw] (1003) to (1002);
\path[draw] (1103) to (1102);
\path[draw] (1013) to (1012);
\path[draw] (0103) to (0102);
\path[draw] (0113) to (0112);
\path[draw] (0013) to (0012);

\path[draw,in=120,out=-120] (1) to (3);
\path[draw] (3) to (2);
\path[draw] (1) to (2);

    \end{tikzpicture}
    \caption{\label{fig:borne}The graph $H_4^3$ -- Extremal case for a graph with a weak $4$-code of size $3$}
  \end{figure}

\bigskip
Light $r$-codes are related to other locating notions:
a \emph{light $1$-code} is a \emph{$1$-locating dominating code} \cite{CSS87} for which we require that only
pairs  of  vertices not in the code are \emph{$1$-separated} by $C$.
The notion of light $r$-codes is a generalization of the notion of metric basis.
A subset $C$ of vertices is a \emph{metric basis} 
if every pair of vertices of the graph 
is separated by a vertex of $C$ for some radius (there is no bound on the radius). 
The \emph{metric dimension} of a graph $G$, denoted by $dim(G)$, is the cardinal of a minimum metric basis. 
A light $r$-code is a metric basis, so  $LC_r(G) \geq dim(G)$.
If $r$ is greater than the diameter of $G$, \ie{} the largest distance between two vertices of $G$, 
then a light $r$-code is exactly a metric basis.  
For a detailed review about metric basis, see \cite{CEJO00}.
As for metric basis, we do not have good bounds of the extremal size of a graph that has light $r$-codes of size $k$.

\bigskip
The optimization problems of finding optimum identifying codes \cite{CHL03} and 
optimum metric bases \cite{KRR96} are NP-complete. Finding optimum light codes is also NP-complete because if $r$ is 
larger than the diameter of the graph, then it is equivalent to metric bases.
Therefore, identifying codes and metric bases have been studied in particular classes of graphs 
(see \eg{} \cite{BCHL05, CGH06,CHHL04,HL02}).

For cycles, although metric bases problem in cycles is not difficult (the dimension of a cycle is $2$), 
the case of identifying codes is not as easy: the complete study
of cycles
has just been finished in \cite{JL09} after numerous contributions
(see \eg{} \cite{BCHL04,GMS06,XTH08}). We focus on the case of weak and light $r$-codes.

\bigskip

In this paper, we give exact value for $\wc{r}$ (Section \ref{sec:weak}) and $LC_r$ (Section \ref{sec:light}) 
for the class of cycles. In weak codes, we assign a radius to each vertex 
to separate it from other vertices whereas we can assign up to 
$r+1$ radii to a vertex with light $r$-codes. We show that $3$ radii per vertex is actually sufficient 
to separate it from all the other vertices.
We adress in Section \ref{sec:2radius} the question of the optimum size of a code requiring only 
$2$ stored radii per vertex.

\

\section{Weak $r$-codes of cycles}\label{sec:weak}
In the following,
we will denote by $\cy_n$ the cycle of size $n$ and by $\{v_0,v_1,\dots,v_{n-1}\}$ the set of its vertices.
We first assume that $n \geq 2r+2$. 

\begin{lemma}\label{lem:blocweak}
Let $S$ be a set of $2r+2$ consecutive vertices on $\cy_n$.
If $C$ is a weak $r$-code of $\cy_n$, then $S$ contains at least two elements of $C$.
\end{lemma}

\begin{proof}
Without loss of generality,
$S=\{v_0,v_1,\ldots, v_{2r+1}\}$. Assume $S$ contains a single element
of the code, say $a=v_i$ , w.l.o.g. $i\leq r$ (see Figure \ref{fig:blocweak}).
\begin{center}
  \begin{figure}[h]
\centering
    \begin{tikzpicture}[scale=0.7]
\fill[black!5,rounded corners=2pt] (0,-0.4) rectangle node[above=14pt] {} (10,0.7);
\path[draw,rounded corners=2pt] (0,-0.4) rectangle node[above=14pt] {\large $S$} (10,0.7);


\path[draw,dashed] (-1,0)--(-0.4,0);
\path[draw] (-0.4,0)--(0.4,0);
\path[draw,dashed] (0.4,0)--(3.8,0);
\path[draw] (3.8,0)--(5.4,0);
\path[draw,dashed] (5.4,0)--(9.4,0);
\path[draw] (9.4,0)--(10.4,0);
\path[draw,dashed] (10.4,0)--(11.2,0);

\draw (3.8,0) circle (4pt) node[below = 10pt] {$x$};
\draw (4.6,0) circle (4pt) node[below = 10pt] {$y$};
\draw (5.4,0) circle (4pt) node[below = 10pt] {$z$};
\draw[fill=white] (-0.4,0) circle (4pt);
\draw[fill=white] (10.4,0) circle (4pt);
\path[fill=white] (2.3,0) circle (4 pt) node[below=10 pt] {$a$};
\draw[fill=white] (0.4,0) circle (4pt) node[above = 0.7pt] {\small $v_0$};
\draw[fill=white] (3.8,0) circle (4pt) node[above = 0.7pt] {\small $v_{r-1}$};
\draw[fill=white] (5.4,0) circle (4pt) node[above = 0.7pt] {\small $v_{r+1}$};
\draw[fill=white] (4.6,0) circle (4pt) node[above = 0.7pt] {\small $v_r$};
\draw[fill] (2.3,0) circle (4pt) node[above = 0.7pt] {\small $v_i$};
\draw[fill=white] (9.4,0) circle (4pt) node[above = 0.7pt] {\small $v_{2r+1}$};
    \end{tikzpicture}    
    \caption{\label{fig:blocweak} Notation of the proof (Lemma \ref{lem:blocweak})}
  \end{figure}
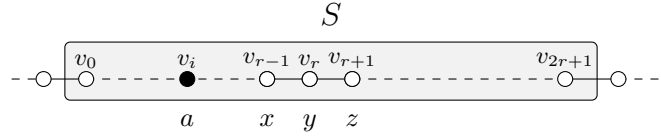
\end{center}
We focus on the
vertices $x=v_{r-1}$, $y=v_r$ and $z=v_{r+1}$. Then, $B_r(y)\subseteq S$ and $B_r(z)\subseteq S$.
Let $t=d(a,y)=r-i$. 

For all $r' \in \lb 0,t-1 \rb$, $B_{r'}(y)\cap C =
B_{r'}(z)\cap C=\emptyset$. For all $r'\in \lb t+1,r \rb$, $B_{r'}(y)\cap C =
B_{r'}(z) \cap C=\{a\}$. Hence $r_y = r_z = t$.
It follows that $B_t(y)\cap C =\{a\}$ must be different from $B_t(x)\cap C$. 
Hence, $B_t(x)\cap C$ must contain an element different from $a$, say $b$. Necessarily, $b \notin S$, this implies
$t=r$ and $z$ is not $r$-dominated, a contradiction.
\end{proof}

\

A first bound of $\wc{r}(\cy_n)$ directly follows from Lemma \ref{lem:blocweak}: 


\begin{corollary}\label{cor:weak}
 Let $C$ be a weak $r$-code of $\cy_n$. Then $|C|\geq \left\lceil n/(r+1)\right\rceil$.
\end{corollary}

\begin{proof}
  In $\cy_n$ there are $n$ different sets $S$ of $2r+2$ consecutives vertices.
If $C$ is a weak $r$-code, by Lemma \ref{lem:blocweak}, there are at least $2$ vertices of the code in each set 
$S$.
Each vertex of the code is counted exactly $2r+2$ times, so $|C|\geq \left\lceil 2n/(2r+2)\right\rceil
=\left\lceil n/(r+1) \right\rceil$.
\end{proof}

\

\noindent In the following, we set $n=(2r+2)p+R$, with $0\leq R \leq
2r+1$ and $p\geq 1$ (by assumption, $n \geq 2r+2$).  
Then Corollary \ref{cor:weak} can be reformulated as: if $C$ is a weak $r$-code of $\cy_n$, then we have
\begin{itemize}[noitemsep]
\item if $R=0$, then $|C|\geq 2p$;
\item if $1 \leq R \leq r+1$, then $|C|\geq 2p+1$;
\item if $r+2 \leq R \leq 2r+1$, then $|C|\geq 2p+2$.
\end{itemize}

Lemmas
\ref{lem:weak1} to \ref{lem:weak12} give some constructive upper
bounds. Moreover, Lemmas \ref{lem:weak1} to \ref{lem:weakdif} provides
exact values of $\wc{r}(\cy_n)$.

\begin{lemma}\label{lem:weak1} 
If $n=(2r+2)p$, then $\cy_n$ has a weak $r$-code
with cardinality $2p=n/(r+1)$; moreover, this code is optimum.
\end{lemma}

\begin{proof}
  We construct the code by repeating the pattern depicted by Figure  \ref{fig:motifweak}. 
  More precisely, let $C=\{v_i \st i\equiv r \ [2r+2]$ or $i\equiv r+1 \ [2r+2]\}$.
  The set $C$ has cardinality $2p$. The set $C$ $r$-dominates all the vertices of $\cy_n$.
  Let $r_{v_k}=r-k$ if $k \in \lb 0,r \rb$ and $r_{v_k}=k-(r+1)$ if $k \in \lb r+1,2r+1 \rb$ 
(the indices of the vertices of $\cy_n$ are taken modulo $2r+2$). 
Then for all pair of vertices $v_k,v_{l}$, $k\neq l$, we have $B_{r_{v_k}}(v_k) \cap C \neq
B_{r_{v_k}}(v_{l}) \cap C$. Hence $C$ is an
$r$-dominating set that $r_{v_k}$-identifies the vertex $v_k$. It
follows that $C$ is a weak $r$-code. This code is optimum by Corollary \ref{cor:weak}.
Figure \ref{fig:weak} gives an example of such a code in $\cy_{12}$.

\end{proof}

\begin{figure}[h]
\begin{center}
\begin{tikzpicture}[scale = 1]
\path[draw] (0,-0.6) -- (0,0.6);
\path[draw] (10,-0.6) -- (10,0.6);

\path[draw] (-1,0)--(1,0);
\path[draw][dashed] (1,0)--(3,0);
\path[draw] (3,0)--(7,0);
\path[draw][dashed] (7,0)--(9,0);
\path[draw] (9,0)--(11,0);
\path[draw, fill] (4.5,0) circle (5 pt);
\path[draw, fill] (5.5,0) circle (5 pt);
\draw[draw, fill=white] (3.5,0) circle (5pt);
\draw[draw, fill=white] (0.5,0) circle (5pt);
\draw[draw, fill=white] (6.5,0) circle (5pt);
\draw[draw, fill=white] (9.5,0) circle (5pt);
\draw[snake=brace,mirror snake] (0.5,-0.5) -- node[below=2pt] {$r$} (3.5,-0.5);
\draw[snake=brace, mirror snake] (6.5,-0.5) -- node[below=2pt] {$r$} (9.5,-0.5);

\end{tikzpicture}
\caption{\label{fig:motifweak} The pattern for a weak $r$-code in the cycles $\cy_{(2r+2)p}$ with $p\geq 1$}
\end{center}
\end{figure}
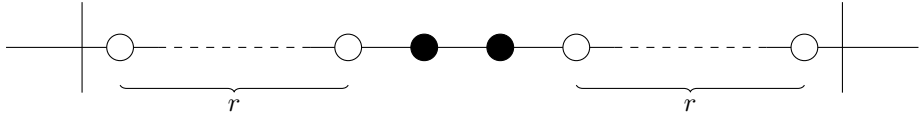

\begin{figure}[h]
    \begin{center}
 \begin{minipage}[c]{.46\linewidth}
\centering
\begin{tikzpicture}[scale=1.4]
 
    \tikzstyle{every node}=[]
\node at (0:1.3) {$v_4$};
\node at (30:1.3) {$v_3$};
\node at (60:1.3) {$v_2$};
\node at (90:1.3) {$v_1$};
\node at (120:1.3) {$v_0$};
\node at (150:1.3) {$v_{11}$};
\node at (180:1.3) {$v_{10}$};
\node at (210:1.3) {$v_9$};
\node at (240:1.3) {$v_8$};
\node at (270:1.3) {$v_7$};
\node at (300:1.3) {$v_6$};
\node at (330:1.3) {$v_5$};

\path[draw,help lines] (0,0) circle (1 cm);
\tikzstyle{every node}=[shape=circle,fill=white,draw=black,minimum
      size=0.5pt,inner
      sep=1.5pt]
\node at (0:1) {};
\node at (30:1) {};
\node at (60:1) {};
\node at (90:1) {};
\node at (120:1) {};
\node at (150:1) {};
\node at (180:1) {};
\node at (210:1) {};
\node at (240:1) {};
\node at (270:1) {};
\node at (300:1) {};
\node at (330:1) {};

\tikzstyle{every node}=[shape=circle,fill=black,draw=black,minimum
      size=0.5pt,inner
      sep=1.5pt]
      \node at (60:1) {};
\node at (30:1) {};
\node at (210:1) {};
\node at (240:1) {};

    \end{tikzpicture}
    \caption{\label{fig:weak}An optimum weak 2-code of $\cy_{12}$}
   \end{minipage} \hfill
   \begin{minipage}[c]{.46\linewidth}
\centering
 \begin{tikzpicture}[scale=1.4]
      
\path[draw,help lines] (0,0) circle (1 cm);
\tikzstyle{every node}=[shape=circle,fill=white,draw=black,minimum
      size=0.5pt,inner
      sep=1.5pt]
\node at (0:1) {};
\node at (27.7:1) {};
\node at (55.4:1) {};
\node at (83.1:1) {};
\node at (110.8:1) {};
\node at (138.5:1) {};
\node at (166.2:1) {};
\node at (193.9:1) {};
\node at (221.5:1) {};
\node at (249.2:1) {};
\node at (277.0:1) {};
\node at (304.6:1) {};
\node at (332.3:1) {};

\tikzstyle{every node}=[]
\node at (0:1.3) {$v_4$};
\node at (27.7:1.3) {$v_3$};
\node at (55.4:1.3) {$v_2$};
\node at (83.1:1.3) {$v_1$};
\node at (110.8:1.3) {$v_0$};
\node at (138.5:1.3) {$v_{12}$};
\node at (166.2:1.3) {$v_{11}$};
\node at (193.9:1.3) {$v_{10}$};
\node at (221.5:1.3) {$v_9$};
\node at (249.2:1.3) {$v_8$};
\node at (277.0:1.3) {$v_7$};
\node at (304.6:1.3) {$v_6$};
\node at (332.3:1.3) {$v_5$};

\tikzstyle{every node}=[shape=circle,fill=black,draw=black,minimum
      size=0.5pt,inner
      sep=1.5pt]
 \node at (27.7:1) {};
\node at (55.4:1) {};    
\node at (221.5:1) {};
\node at (249.2:1) {};
\node at (138.5:1) {};
    \end{tikzpicture}
    \caption{\label{fig:test}An optimum weak 2-code of $\cy_{13}$}
   \end{minipage}
\end{center}
\end{figure}

\bigskip

We can easily extend this construction to the general case:

\begin{lemma}\label{lem:weakex}
  If $R=1$, then $\cy_n$ has a weak $r$-code with $2p+1$ elements. 
  If $2\leq R\leq 2r+1$, then $\cy_n$ has a weak $r$-code with $2p+2$ elements.
  These codes are optimum for $R=1$ or $R \geq r+2$.
\end{lemma}

\begin{proof}
  Let $R=1$ and $C=\{v_i \st i\equiv r \ [2r+2]$ or $i\equiv r+1 \ [2r+2] \}\cup\{v_{n-1}\}$.
  Then $C$ is a weak $r$-code of $\cy_n$ and $|C|=2p+1$. (See Figure \ref{fig:test}.)

\noindent Assume now that $R\geq 1$ and take for code $C=\{v_i \st
i\equiv r [2r+2]$ or $i\equiv r+1 [2r+2] \}$ if $R \geq r+2$
and $C=\{v_i \st i\equiv r [2r+2]$ or $i\equiv r+1 [2r+2] \}\cup\{v_{n-2},v_{n-1}\}$ otherwise.
Then $C$ is a weak $r$-code of $\cy_n$.
\end{proof}

\

In some cases, the aforementioned codes are not optimum:

\begin{lemma}\label{lem:weak12}
  If $(r,R)=(1,2)$, then $\cy_n$ has an optimum weak $1$-code of cardinality $2p+1$.
  If $(r,R)=(2,2)$, then $\cy_n$ has an optimum weak $2$-code of cardinality $2p+1$.
\end{lemma}

Figure \ref{fig:weakr1R2} (\emph{resp.} Figure \ref{fig:weakr2R2}) shows an example of an optimum weak $r$-code for $(r,R)=(1,2)$ (\emph{resp. $(r,R)=(2,2)$}).

\begin{figure}[h]
  \begin{center}
    \begin{minipage}[c]{.46\linewidth}
      \centering
      \begin{tikzpicture}[scale=1.4]
        
        \path[draw,help lines] (0,0) circle (1 cm);
        \tikzstyle{every node}=[shape=circle,fill=white,draw=black,minimum
        size=0.5pt,inner
        sep=1.5pt]
        \node at (0:1) {};
        \node at (36:1) {};
        \node at (72:1) {};
        \node at (108:1) {};
        \node at (144:1) {};
        \node at (180:1) {};
        \node at (216:1) {};
        \node at (252:1) {};
        \node at (288:1) {};
        \node at (324:1) {};
        
        \tikzstyle{every node}=[]
        \node at (0:1.3) {$v_4$};
        \node at (36:1.3) {$v_3$};
        \node at (72:1.3) {$v_2$};
        \node at (108:1.3) {$v_1$};
        \node at (144:1.3) {$v_0$};
        \node at (180:1.3) {$v_9$};
        \node at (216:1.3) {$v_8$};
        \node at (252:1.3) {$v_7$};
        \node at (288:1.3) {$v_6$};
        \node at (324:1.3) {$v_5$};
        
        \tikzstyle{every node}=[shape=circle,fill=black,draw=black,minimum
        size=0.5pt,inner
        sep=1.5pt]
        \node at (0:1) {};
        \node at (72:1) {};
        \node at (144:1) {};
        \node at (216:1) {};
        \node at (288:1) {};
    \end{tikzpicture}
    \caption{\label{fig:weakr1R2}An optimum weak 1-code of $\cy_{10}$}
  \end{minipage} \hfill
  \begin{minipage}[c]{.46\linewidth}
    \centering
    \begin{tikzpicture}[scale=1.4]
      
      \path[draw,help lines] (0,0) circle (1 cm);
\tikzstyle{every node}=[shape=circle,fill=white,draw=black,minimum
      size=0.5pt,inner
      sep=1.5pt]
\node at (0:1) {};
\node at (45:1) {};
\node at (90:1) {};
\node at (135:1) {};
\node at (180:1) {};
\node at (225:1) {};
\node at (270:1) {};
\node at (315:1) {};

\tikzstyle{every node}=[]
\node at (0:1.3) {$v_4$};
\node at (45:1.3) {$v_3$};
\node at (90:1.3) {$v_2$};
\node at (135:1.3) {$v_1$};
\node at (180:1.3) {$v_0$};
\node at (225:1.3) {$v_{7}$};
\node at (270:1.3) {$v_{6}$};
\node at (315:1.3) {$v_{5}$};
\tikzstyle{every node}=[shape=circle,fill=black,draw=black,minimum
      size=0.5pt,inner
      sep=1.5pt]
\node at (90:1) {};
\node at (180:1) {};
\node at (270:1) {};

    \end{tikzpicture}
    \caption{\label{fig:weakr2R2}An optimum weak 2-code of $\cy_{8}$}
   \end{minipage}
\end{center}
\end{figure}
\begin{proof} 
For $(r,R)=(1,2)$, the set $C=\left \{v_{i} \st i \equiv 0 [2] \right \}$ is a weak $1$-code: 
  each vertex $x$ in the code is $0$-identified by $C$ and each vertex $x$ 
not in the code is $1$-identified by $C$. 
For $(r,R)=(2,2)$, the set $C=\{v_i \st i\equiv 0 \ [6]$ or $i\equiv 2 \ [6]\}$ 
is a weak $2$-code. The optimality of these codes is shown by Corollary \ref{cor:weak}.
\end{proof}\\
\

The next lemma shows that the lower bound of Corollary \ref{cor:weak}
is not sharp for $2\leq R \leq r+1$ and $(r,R)\neq (1,2)$ or $(2,2)$, this
implies that in these cases, codes of Lemma \ref{lem:weakex} are
optimum.

\begin{lemma}\label{lem:weakdif}
  If $2\leq R\leq r+1$ and $(r,R)\neq (1,2)$ or $(2,2)$, 
then $\cy_n$ does not have a weak $r$-code of cardinality $2p+1$.
\end{lemma}

\begin{proof}
\def\ca{c_0}
\def\cb{c_1}
\def\cc{c_2}
\def\cd{c_3}
  Assume that there is a weak $r$-code $C$ of $\cy_n$ of cardinality $2p+1$. First, observe:
  \begin{enumerate}
\item[(O.1)] In a set of $R$ consecutive vertices of $\cy_n$, 
there must be at most one vertex of $C$. Otherwise, in the rest of
$\cy_n$, there 
are at most $2p-1$ vertices of the code in a set of $(2r+2)p$ consecutive vertices 
which contradicts Lemma \ref{lem:blocweak}.
In particular there is no pair of consecutive vertices of $C$.

\item[(O.2)] For similar reasons, in a set of $2r+2+R$ consecutive vertices of $\cy_n$, 
there must be at most $3$ vertices of $C$.
\end{enumerate}

Let $M$ be the maximum size of a set of consecutive vertices not in $C$ 
and let $S_M$ be a set of $M$ consecutive vertices not in $C$. 
We know by (O.1) that $M\geq R-1$.
Moreover $M>1$; indeed, if $M=1$, then $R=2$ and the code is exactly one vertex over 2, so $|C|=\frac{n}{2}=2p+1$, $n=4p+2$ 
and $(r,R)=(1,2)$.

Let us denote $\cb$ and $\cc$ the two elements of the code bounding $S_M$, 
let $S_1$ and $S_2$ be the two maximal sets of consecutive vertices 
not in $C$ who are before $\cb$ and after $\cc$, and finally $\ca$ and $\cd$ the two vertices of the code who are before $S_1$ and after $S_2$ (see Figure \ref{fig:not}).

\vspace{-1em}
\begin{center}
  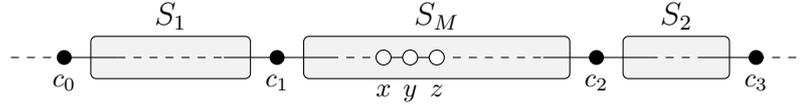
\begin{figure}[h]
\centering
    \begin{tikzpicture}[scale=0.7]

\fill[black!5,rounded corners=2pt] (0.5,-0.4) rectangle node[above=7pt] {} (3.5,0.4) ;
\path[draw,rounded corners=2pt] (0.5,-0.4) rectangle node[above=7pt] {\large $S_1$} (3.5,0.4) ;

\fill[black!5,rounded corners=2pt] (4.5,-0.4) rectangle node[above=7pt] {} (9.5,0.4);
\path[draw,rounded corners=2pt] (4.5,-0.4) rectangle node[above=7pt] {\large $S_M$} (9.5,0.4);

\fill[black!5,rounded corners=2pt] (10.5,-0.4) rectangle node[above=7pt] {} (12.5,0.4);
\path[draw,rounded corners=2pt] (10.5,-0.4) rectangle node[above=7pt] {\large $S_2$} (12.5,0.4);

\path[draw,dashed] (-1,0)--(0,0);
\path[draw] (0,0)--(1,0);
\path[draw,dashed] (1,0)--(3,0);
\path[draw] (3,0)--(5,0);
\path[draw,dashed] (5,0)--(6,0);
\path[draw] (6,0)--(7,0);
\path[draw,dashed] (7,0)--(9,0);
\path[draw] (9,0)--(11,0);
\path[draw,dashed] (11,0)--(12,0);
\path[draw] (12,0)--(13,0);
\path[draw,dashed] (13,0)--(14,0);


\path[fill] (0,0) circle (4 pt) node[below=3 pt] {$\ca$};
\fill (4,0) circle (4 pt) node[below =3 pt ] {$\cb$};
\fill (10,0) circle (4 pt) node[below =3pt] {$\cc$};
\fill (13,0) circle (4 pt) node[below=3pt] {$\cd$};
\draw[fill=white] (6,0) circle (4pt) node[below = 7pt] {$x$};
\draw[fill=white] (6.5,0) circle (4pt) node[below = 7pt] {$y$};
\draw[fill=white] (7,0) circle (4pt) node[below = 7pt] {$z$};

\end{tikzpicture}
\caption{\label{fig:not}Notation of the proof (Lemma \ref{lem:weakdif})}
\end{figure}
\end{center}

\vspace{-2.5em}

\begin{itemize}[noitemsep]
\item Observe that $p\geq 1$, so $C$ has cardinality at least $3$ and
observe by (O.1) that $S_1$ and $S_2$ are not empty.  Hence, the elements
$\cb$, $\cc$, $\cd$ may be supposed distincts and so on for elements
$\ca$, $\cb$ and $\cc$, but note that $\ca$ and $\cd$ may denote the
same vertex.

\item Observe by (O.1) that $|S_1|\geq R-1$, $|S_2|\geq R-1$, $M\geq R-1$.
Let us denote $S$ the set $S_1\cup \{\cb \}\cup S_M \cup \{\cc \} \cup S_2$. 

\item Observe that $\vert S \vert \geq 2r+3$. Indeed, if $\ca$ and $\cd$ are different vertices, 
then $\{\ca\}\cup S \cup \{\cd\}$ is a set with $4$ vertices of the code, 
so, by (O.2) $|S|+2 > 2r+2+R \geq 2r+4$. 
If $\ca$ and $\cd$ denote the same vertex, then $S\cup\{\cd \}=V(\cy_n)$, $p=1$ and $\vert S \vert =n-1=2r+1+R \geq 2r+3$.
\end{itemize}

\noindent So there are three consecutive vertices $x,y,z$ in $S$ such that 
$\left \{  B_r(x)\cup B_r(y) \cup B_r(z) \right \} \cap C \subseteq \{\cb, \cc\}$ and $y\in S_M$.

\vspace{0.5em}


To separate $y$ and $x$, $r_y$ must be $d(x,\cb)$ or $d(y,\cc)$.  To
separate $y$ and $z$, $r_y$ must be $d(y,\cb)$ or $d(z,\cc)$.
Therefore, either $r_y=d(x,\cb)=d(z,\cc)$, or $r_y=d(y,\cc)=d(y,\cb)$.
In all cases, $M$ is odd and $y$ is the middle element of $S_M$, so
$d(y,\cb)=d(y,\cc)$. As $M \neq 1$ then $M\geq 3$ and $(x,z) \in
S_{M}\times S_{M}$.

\vspace{0.5em}
Let $d_y$ denote $d(y,\cb)$ in the following. Let $w$ be the vertex just before $x$. Then $B_r(w) \cap C \subseteq\{\ca
,\cb,\cc\}$.  To separate $x$ from $y$, $r_x$ must be $d(y,\cc)=d_y$
or $d(x,\cb)=d_y-1$.  To separate $x$ from $w$, $r_x$ must be
$d(w,\cb)=d_y-2$ or $d(x,\cc)=d_y+1$ or $d(w,\ca)$.  Necessarily, we
have $r_x=d(w,\ca)$. This implies $d(w,\ca)=r$ because
$d(w,\ca)=d(x,\ca)-1 \geq r$ and $r_{x}\leq r$. Since $d_y\leq r$ and
$r_{x}=d_y$ or $r_{x}=d_y-1$. It follows $r_x=d_y=r$. Therefore $M
=2r-1$, $\vert S_1 \vert = 1$, and finally $R=2$.  With similar
arguments for $z$, we obtain the situation depicted by Figure
\ref{fig:simpl}.

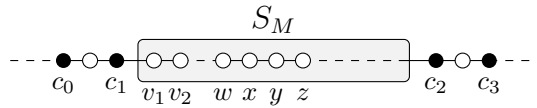
\begin{figure}[h]
\centering
    \begin{tikzpicture}[scale=0.7]

\fill[black!5,rounded corners=2pt] (1.4,-0.4) rectangle node[above=7pt] {} (6.5,0.4);
\draw[rounded corners=2pt] (1.4,-0.4) rectangle node[above=7pt] {\large $S_M$} (6.5,0.4);

\path[draw,dashed] (-1,0)--(0,0);
\path[draw] (0,0)--(2.2,0);
\path[draw,dashed] (2.2,0)--(3,0);
\path[draw] (3,0)--(4.5,0);
\path[draw,dashed] (4.5,0)--(6.3,0);
\path[draw] (6.3,0)--(8,0);
\path[draw,dashed] (8,0)--(9,0);


\draw[fill=white] (0.5,0) circle (4 pt) node[above=4pt] {} (3.5,0.3) ;
\draw[fill=white] (7.5,0) circle (4 pt) node[above=4pt] {} (3.5,0.3) ;

\path[fill] (0,0) circle (4 pt) node[below=3 pt] {$\ca$};
\fill (1,0) circle (4 pt) node[below =3 pt ] {$\cb$};
\fill (7,0) circle (4 pt) node[below =3pt] {$\cc$};
\fill (8,0) circle (4 pt) node[below=3pt] {$\cd$};
\draw[fill=white] (1.7,0) circle (4pt) node[below = 7pt] {$v_1$};
\draw[fill=white] (2.2,0) circle (4pt) node[below = 7pt] {$v_2$};
\draw[fill=white](3,0) circle (4pt) node[below = 7pt] {$w$};
\draw[fill=white](3.5,0) circle (4pt) node[below = 7pt] {$x$};
\draw[fill=white] (4,0) circle (4pt) node[below = 7pt] {$y$};
\draw[fill=white] (4.5,0) circle (4pt) node[below = 7pt] {$z$};

    \end{tikzpicture}
    \caption{\label{fig:simpl}The sets $S_1, S_2$ and $S_M$ after some deductions}
  \end{figure}

Consider $(r,R)\neq (1,2)$ or $(2,2)$ and $R=2$, then $r\geq 3$ and so $M \geq 5$. 
Let $v_1$ and $v_2$ be the two consecutive vertices in $S_M$ 
following $\cb$ (see Figure~\ref{fig:simpl}). 
We have $d(v_2,\cc)=M - 1 > r$ and $d(v_1,\cc)>r$ so $v_1$ and $v_2$ can only be separated 
by elements of the code on the left of $v_1$ and $v_2$. Let $r_{v_1}$ be the radius that identifies $v_1$.
There must be an element of the code at distance 
exactly $r_{v_1}$ of $v_1$ to separate $v_1$ and $v_2$, and for similar reasons, there must be an element of the code 
at distance $r_{v_1}+1$ of $v_1$ to separate $v_1$ from $\cb$. 
This implies that two elements of the code are consecutives vertices in $\cy_n$, which contradicts (O.1). 
\end{proof}

\vspace{\baselineskip}

We are now able to compute $\wc{r}(\cy_n)$ for all $n\geq 2r+2$. 
Our results are summarized in the following theorem:

\begin{theorem}\label{theo:weakpart1}
  Let $r$ be an integer and $n=(2r+2)p + R$, with $0\leq R \leq 2r+1$ and $p\geq 1$, we have:
  \begin{enumerate}[noitemsep,label=\roman*)]
  \item  if $R=0$, then $\wc{r}(\cy_n)=2p$,
  \item  if $R=1$ or if $r\leq 2$ and $R=2$, then $\wc{r}(\cy_{n})=2p+1$,
  \item  otherwise, $R\geq 2$ and 
$(r,R)\neq (1,2)$ or $(2,2)$, then $\wc{r}(\cy_{n})=2p+2$.
  \end{enumerate}

\end{theorem}



The following lemma completes the study for the small cases:

\begin{lemma}\label{lem:small}
  Let $r$ and $n$ be integers with $3\leq n\leq 2r+1$, 
then $\wc{r}(\cy_n)=2$.
\end{lemma}

\begin{proof}
The code cannot be a single vertex, otherwise its two neighbors are
not $i$-separated for any $i$, so $\wc{r}(\cy_n)\geq 2$.  Two adjacent
vertices form a weak $r$-code for any $r$, so $\wc{r}(\cy_n)=2$. Note
that if $n$ is odd, the antipodal vertex to the code in the cycle is
identified by the empty set.
\end{proof}

\section{Light $r$-codes of cycles}\label{sec:light}

We now study light $r$-codes of the cycle $\cy_n$.
In this section, we will first assume that $n\geq 3r+2$ and we will study the small values of $n$ 
at the end of the section.

\begin{lemma}\label{lem:bloc1}
  Let $C$ be a light $r$-code of $\cy_n$ and $c$ an element of $C$.
There is another element of the code $C$ at distance 
at most
$r+1$ of $c$.
\end{lemma}

\begin{proof}
Let $x$ and $y$ be the neighbors of $c$. As $C$ is a light $r$-code,
there is an integer $r_{xy}$ such that $0\leq r_{xy} \leq r$ and
$B_{r_{xy}}(x) \cap C \neq B_{r_{xy}}(y) \cap C$.  There consequently
exists a vertex $c'\in C$ such that, w.l.o.g., $c'\in B_{r_{xy}}(x)$
and $c'\notin B_{r_{xy}}(y)$.  Moreover, $c\neq c'$ because
$d(x,c)=d(c,y)=1$. It follows that $d(c^\prime,c) \leq d(c^\prime,x)+d(x,c) \leq
r_{xy}+ 1 \leq r +1$.
\end{proof}

\begin{lemma}\label{lem:bloc2}
Let $S$ be a set of $3r+2$ consecutive vertices on $\cy_n$.
If $C$ is a light $r$-code of $\cy_n$, then $S$ contains at least two elements of $C$.
\end{lemma}

\begin{proof}
  Let $C$ be a light $r$-code of $\cy_n$. 
Let us assume there is a set $S$ of $3r+2$ consecutive vertices of $\cy_n$
containing only one element $c$ of $C$.
w.l.o.g., we denote $S=\{v_0,v_1, \dots ,v_{3r+1}\}$ and $c=v_i$ with $i<2r$.
By Lemma \ref{lem:bloc1}, there is an element $c'$ at distance at most
$r+1$ of $c$. But $c' \notin S$ so necessarily, $c' \in \{v_{-1},v_{-2}, \ldots ,v_{-(r+1)}\}$
and $i\leq r$. Then $v_{2r+1}$ is not $r$-dominated by any element of $C$,
a contradiction.
\end{proof}

\bigskip

\noindent It follows from Lemma \ref{lem:bloc2}:

\begin{corollary}\label{cor:light}
Let $C$ be a light $r$-code of $\cy_n$. Then $|C|\geq \left\lceil 2n/(3r+2) \right\rceil$.
\end{corollary}

\noindent In the following, let $n=(3r+2)p+R$ with $0\leq R \leq 3r+1$ and $p\geq 1$ (by assumption, $n\geq 3r+2$). Then Corollary \ref{cor:light} can be reformulated as: 
if $C$ is a light $r$-code of $\cy_n$, then we have
\begin{itemize}[noitemsep]
\item if $R=0$, then $\vert C \vert \geq 2p$,
\item if $0 < 2R \leq 3r+2$, then $\vert C \vert \geq 2p+1$,
\item otherwise, $2R > 3r+2$, and $\vert C \vert \geq 2p+2$.
\end{itemize}

\

We want to exhibit some optimum codes.
\

\begin{lemma}\label{lem:light1}
If  $n=(3r+2)p$, then $\cy_n$ has a light $r$-code
with cardinality $2p$. Moreover this code is optimum.
\end{lemma}

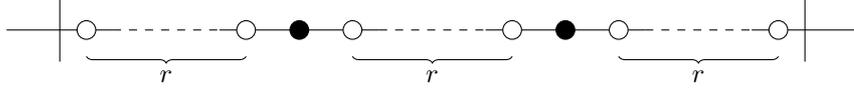
\begin{figure}[h]
\begin{center}
\begin{tikzpicture}[scale = 0.7]
\path[draw] (0,-0.6) -- (0,0.6);
\path[draw] (14,-0.6) -- (14,0.6);

\path[draw] (-1,0)--(1,0);
\path[draw][dashed] (1,0)--(3,0);
\path[draw] (3,0)--(6,0);
\path[draw][dashed] (6,0)--(8,0);
\path[draw] (8,0)--(11,0);
\path[draw][dashed] (11,0)--(13,0);
\path[draw] (13,0)--(15,0);

\path[draw, fill] (4.5,0) circle (5 pt);
\path[draw, fill] (9.5,0) circle (5 pt);
\draw[fill=white] (3.5,0) circle (5pt);
\draw[fill=white] (0.5,0) circle (5pt);
\draw[fill=white] (10.5,0) circle (5 pt);
\draw[fill=white] (13.5,0) circle (5 pt);
\draw[fill=white] (5.5,0) circle (5pt);
\draw[fill=white] (8.5,0) circle (5pt);
\draw[snake=brace,mirror snake] (0.5,-0.5) -- node[below=2pt] {$r$} (3.5,-0.5);
\draw[snake=brace, mirror snake] (5.5,-0.5) -- node[below=2pt] {$r$} (8.5,-0.5);
\draw[snake=brace, mirror snake] (10.5,-0.5) -- node[below=2pt] {$r$} (13.5,-0.5);

\end{tikzpicture}
\caption{\label{fig:motiflight} The pattern $S$ for a light $r$-code in the cycles $\cy_{(3r+2)p}$ with $p \geq 1$}
\end{center}
\end{figure}

\begin{proof}
We construct the code by repeating the pattern $S$ depicted by Figure \ref{fig:motiflight}. 
More precisely, let $C=\{v_i \st i\equiv r \ [3r+2] \mbox{ or }  i\equiv 2r+1 \ [3r+2] \}$. 
\noindent Set  $C$ is a $r$-dominating set of size $2p$ and we just need to check that every pair 
of vertices is separated by $C$ for some radius in $\lb 0,r \rb$. It is sufficient to 
prove it for all pairs $(v_i,v_j)$ in the pattern $S$, \emph{i.e.} with $(i,j) \in \lb 0,3r+1 \rb\times \lb 0,3r+1 \rb$.
W.l.o.g. we study the case $i < j$, and we define $r_{ij}$ as follows:

\begin{itemize}[noitemsep]
\item if $j\leq r$, then $r_{ij}=r-j$;
\item if $i\leq r < j$, then $r_{ij}=|(2r+1)-j|$;
\item if $r<i\leq 2r$, then $r_{ij}=i-r$;
\item if $i\geq 2r+1$, then $r_{ij}=i-(2r+1)$. 
\end{itemize}
 
\noindent Then, $0\leq r_{ij} \leq r$ and it is easy to check that 
$(v_i,v_j)$ is $r_{ij}$-separated by $C$.
So $C$ is a light $r$-code of $\cy_n$ with cardinality $2p$.
This code is optimum by Corollary \ref{cor:light}.
\end{proof}

\bigskip

\noindent We generalize this construction:

\begin{lemma}\label{lem:lightex}
 If $1 \leq R\leq r+1$, then $\cy_n$ has a light $r$-code of cardinality $2p+1$.
  If $R>r+1$, then $\cy_n$ has a light $r$-code of cardinality $2p+2$.
\end{lemma}

\begin{proof}
Consider the three following cases: (1) $R\in \lb 1,r+1 \rb$, (2) $R\in
\lb r+2,2r+2 \rb$, and (3) $R \in \lb 2r+3,3r+1 \rb$. For each case, we define the code $C$ as:
\begin{itemize}[noitemsep]
\item[(1)] $C=\{v_i \st i< (3r+2)p,  i\equiv r \ \lb 3r+2 \rb$ or $i\equiv 2r+1 \ \lb 3r+2 \rb\} \cup \{v_{(3r+2)p}\}$
\item[(2)] $C=\{v_i \st i< (3r+2)p,  i\equiv r \ \lb 3r+2 \rb$ or $i\equiv 2r+1 \ \lb 3r+2 \rb\} \cup \{v_{(3r+2)p},v_{(3r+2)p+r}\}$
\item[(3)] $C=\{v_i \st i< (3r+2)p, i\equiv r\  \lb 3r+2 \rb$ or $i\equiv 2r+1 \ \lb 3r+2 \rb\} \cup \{v_{(3r+2)p+r},v_{(3r+2)p+2r}\}$
\end{itemize}
These sets are light $r$-codes of cardinality $2p+1$,
$2p+2$ and $2p+2$, respectively.
\end{proof}

\begin{lemma}\label{lem:lightdi}
 If $R>r+1$, then $\cy_n$ has no light $r$-code of cardinality $2p+1$.
\end{lemma}

\begin{proof}
Assume that there is a code $C$ of cardinality $2p+1$. 
First observe that in a set $S$ of $R$ consecutive vertices, there is at most one element of the code $C$.
Otherwise, there will be only $2p-1$ elements of the code in the rest 
of the cycle which can be divided in $p$ disjoint sets of size
$3r+2$. One of this set will have only one element of the code,
a contradiction by Lemma \ref{lem:bloc2}.

\noindent Now, take an element $c$ of the code $C$, by Lemma \ref{lem:bloc1} 
there is a vertex $c'$ of the code at distance $d\leq r+1$ of $c$. 
Take the set $S$ of all vertices between $c$ and $c'$, $c$ and $c'$ included. 
$S$ has cardinality at most $r+2\leq R$ and has two vertices of $C$, a contradiction.

\end{proof}

Our results are summarized in the following theorem:

\begin{theorem}\label{thm:light}
    Let $r$ be an integer and $n=(3r+2)p +R$, with $0\leq R < 3r+2$, and $p\geq 1$, we have:    
\begin{enumerate}[noitemsep,label=\roman*)]
\item if $R=0$, then $LC_{r}(\cy_{n})=2p$;
\item if $R\leq r+1$, then $LC_{r}(\cy_{n})=2p+1$;
\item otherwise, $R>r+1$ and then $LC_{r}(\cy_{n})=2p+2$.
\end{enumerate}
\end{theorem}

\noindent Theorem \ref{thm:light}-\emph{i} (\emph{resp.} \ref{thm:light}-\emph{ii}, \ref{thm:light}-\emph{iii}) follows from 
Lemma \ref{lem:light1}  
(\emph{resp.} Corollary \ref{cor:light} and Lemma \ref{lem:lightex}, and from 
Lemmas \ref{lem:lightex} and \ref{lem:lightdi}).

\bigskip

The next lemma completes the study for the small values of $n$:

\begin{lemma}
  Let $r$ and $n$ be integers with $3\leq n\leq 3r+1$, then $LC_r(\cy_n)=2$.
\end{lemma}

\begin{proof}
A light $r$-code cannot be a single vertex otherwise the neighbors of
the element of the code are not $i$-separated for any $i$.
Two adjacent vertices form a light $r$-code for any $n\leq 2r+2$. For $n>2r+2$, take two vertices at distance $r+1$.
\end{proof}

\bigskip

With light $r$-codes, we can assign up to $r+1$ radii to a vertex to separate it 
from all the other vertices. 
Actually, for cycles, we just need three radii:

\begin{proposition} \label{prop:light3}
Let $C$ be a light $r$-code of $\cy_n$ and $x$ be a vertex of $\cy_n$. Assume that $n> 2r+1$. 
There is a subset  $R_x$ of $\lb 0,r \rb$ of size at most $3$ such that for all other vertices $y$ of $\cy_n$, there is $r_{xy}\in R_x$ such that $B_{r_{xy}}(x)\cap C \neq B_{r_{xy}}(y)\cap C$.
\end{proposition}

\begin{proof}
Without loss of generality, we can assume that $x=v_0$.

Assume first that there exist two vertices of the code, say $a=v_i$ and $b=v_j$, such that $-r\leq i \leq 0
 \leq j \leq r$ (if $x \in C$, then we have $a=b=x$). Thus $R_x=\{d(x,a),d(x,b)\}$ separates $x$ from all the other vertices:
vertices $x$ and $v_k$ are separated for radius $d(x,a)$ if $0<k<n/2$ and for radius $d(x,b)$ if $-n/2<k<0$.

Otherwise, let $a=v_i$ be the element of the code closest to $x$. We can assume that $0<i\leq r$. 
By Lemma \ref{lem:bloc1} we know that there exists another element of the code $b=v_j$ such that $i<j$ and $j-i\leq r+1$.
Then $x$ is separated from all vertices not in $B_i(a)$ by radius $i$, and from all vertices in $B_{i-1}(a)$ by radius $i-1$.
It remains one vertex, $v_{2i}$, that is separated from $x$ for radius $d(v_{2i},b)\leq r$. 
Finally the three radii $i$, $i-1$, $d(v_{2i},b)$ are enough to separate $x$ 
from all vertices.
\end{proof}

\vspace{\baselineskip}

This proposition leads to the following question: what is the size of an optimum light 
$r$-code on $\cy_n$ 
that need to assign only $2$ radii to each vertex? We solve this question in the next section.

\section{Codes with 2 radii}\label{sec:2radius}

A \emph{$(2,\lb 0,r \rb)$-code} $C$ of a graph $G$ is a subset of vertices of $G$ that $r$-dominates every vertex 
and such that for each vertex $x$, we can assign a set $R_{x}=\{r_x,r_x'\}$ 
of integers in $\lb 0,r \rb$ such that every pair of distinct vertices $(x,y)$ is $r_x$ or $r'_x$-separated by $C$.

\begin{lemma}\label{lem:2codeex}
Let $k=\left \lfloor (r+1)/3 \right \rfloor$ and $s=3r-k+2$. If $s$ divides $n$, 
then the code defined by repeating the pattern $S$ depicted by 
Figure \ref{fig:motif2code} is a $(2,\lb 0,r \rb)$-code of $\cy_n$.
\end{lemma}

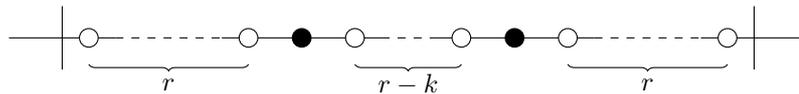
\begin{figure}[h]
\begin{center}
\begin{tikzpicture}[scale = 0.7]
\path[draw] (0,-0.6) -- (0,0.6);
\path[draw] (13,-0.6) -- (13,0.6);

\path[draw] (-1,0)--(1,0);
\path[draw][dashed] (1,0)--(3,0);
\path[draw] (3,0)--(6,0);
\path[draw][dashed] (6,0)--(7,0);
\path[draw] (7,0)--(10,0);
\path[draw][dashed] (10,0)--(12,0);
\path[draw] (12,0)--(14,0);

\path[draw, fill] (4.5,0) circle (5 pt); 
\path[draw, fill] (8.5,0) circle (5 pt);
\draw[fill=white] (3.5,0) circle (5pt);
\draw[fill=white] (0.5,0) circle (5pt);
\draw[fill=white] (9.5,0) circle (5 pt);
\draw[fill=white] (12.5,0) circle (5 pt);
\draw[fill=white] (5.5,0) circle (5pt);
\draw[fill=white] (7.5,0) circle (5pt);
\draw[snake=brace,mirror snake] (0.5,-0.5) -- node[below=2pt] {$r$} (3.5,-0.5);
\draw[snake=brace, mirror snake] (5.5,-0.5) -- node[below] {$r-k$} (7.5,-0.5);
\draw[snake=brace, mirror snake] (9.5,-0.5) -- node[below=2pt] {$r$} (12.5,-0.5);

\end{tikzpicture}
\caption{\label{fig:motif2code} The pattern $S$ for a $(2,\lb 0,r \rb)$-code of the cycle $\cy_n$ with $n$ multiple of $s$ (\emph{cf.} Lemma \ref{lem:2codeex})}
\end{center}
\end{figure}

\begin{proof}
\def\ca{c_1}
\def\cb{c_2}
We focalize on a pattern $S$. Denote by $\ca$ and $\cb$ the two
vertices of the code of $S$ and assume that $\ca=v_{0}$. Then
$\cb=v_{r-k+1}$ and the vertices of $S$ are the vertices between
$v_{-r}$ and $v_{2r-k+1}$. Partition the vertices of $S$ in five
subsets: $A_{1}=\{v_{-r}, \ldots, v_{-k-1}\}$, $A_{2}=\{v_{-k}, \ldots, v_{-1}\}$, $A_{3}=\{v_{0}, \ldots, v_{r-k+1}\}$, $A_{4}=\{v_{r-k+2}, \ldots, v_{r+1}\}$ and $A_{5}=\{v_{r+2}, \ldots, v_{2r-k+1}\}$. If $r=1$, then $A_2$ and $A_4$ are empty; if $r=0$, then $A_3$ is non empty and the other sets are empty.
Let $x$ be a vertex of $S$. Let $R_{x}$ the set of radii associated to
$x$:
\begin{itemize}[noitemsep]
\item if $x\in A_1$, then set $R_x=\{d(x,\ca), d(x,\ca)-1\}$;
\item if $x \in A_{2}$, then set $R_x=\{d(x,\ca), d(x,\cb)-1\}$;
\item if $x \in A_{3}$, then set $R_x=\{d(x,\ca), d(x,\cb)\}$; 
\item if $x \in A_{4}$, then set $R_x=\{d(x,\ca)-1, d(x,\cb)\}$;
\item if $x \in A_{5}$, then set $R_x = \{d(x,\cb), d(x,\cb)-1\}$.
\end{itemize}

One can check that $R_{x} \subset \lb 0,r \rb$ in all cases.  By symmetry,
we just need to check that every vertex $x$ of $A_{1}\cup A_{2}\cup
A_{3}$ is separated from all the other vertices for a radius in
$R_{x}$.

\vspace{.5em}
    
If $x \in A_{1}$, then $x$ is separated from the vertices not in
$B_{d(x,\ca)}(\ca)$ for radius $d(x,\ca)$ and from the vertices in
$B_{d(x,\ca)-1}(\ca)$ for radius $d(x,\ca)-1$. Remains the vertex $y$
at distance $d(x,\ca)$ of $\ca$. If $x=v_{-i}$, with $k+1 \leq i \leq
r$, then $y=v_{i}$ and $d(y,\cb)=r-k+1-i \leq r-2k \leq k+1 \leq
d(x,\ca)$ by definition of $k$. Notice that $d(x,\cb)>d(x,\ca)$, so
$x$ and $y$ are separated for radius $d(x,\ca)$.

\vspace{.5em}    
    
If $x \in A_{2}$, then $x$ is separated from the vertices not in
$B_{d(x,\ca)}(\ca)$ for radius $d(x,\ca)$ and from the vertices in
$B_{d(x,\cb)-1}(\ca)$ for radius $d(x,\cb)-1$. That covers all the
vertices of the cycle.
    
\vspace{.5em}

One can check by the same kind of arguments that $x \in A_{3}$ is also
separated from all the other vertices for $d(x,\ca)$ or $d(x,\cb)$.
\end{proof}

\begin{lemma}
Let $C$ be a $(2,\lb 0,r \rb)$-code of $\cy_n$. Let $S$ be a set of $s=3r-k+2$ vertices with $k=\lfloor (r+1)/3\rfloor$.
Then $S$ contains at least two vertices of $C$.
\end{lemma}

\begin{proof}
  For $r=0$, the lemma is true as all the vertices must be
  $0$-dominated. The lemma is also true for $r=1$, as a $(2,\lb 0,1
  \rb)$-code is a light $1$-code. Now, let $r \geq 2$. Notice that
  $3r-k+2 > 2r$, thus $S$ contains at least one vertex of $C$. By
  contradiction, assume that $S$ contains only one vertex $c$ of $C$,
  and w.l.o.g. assume $c=v_{0}$.
Let $v_{-a}$ be the first
  vertex of $S$ and $v_{b}$ be the last vertex of $S$, $a+b=3r-k+1$.
  We can assume that $a\leq b$. $C$ is also a light $r$-code so by
  Lemma \ref{lem:bloc1} $a\leq r$, then $b\geq 2r-k+1$. $C$ is
  $r$-dominating so $b\leq 2r$, and then $a\geq r-k+1$.  We have
  $B_{r}(v_k)\cap C=B_{r}(v_{k-1})\cap C=B_{r}(v_{k+1})\cap C=\{c\}$
  because $d(v_k,v_{-a})=a+k\geq r+1$ and $d(v_k,v_{b})=b-k\geq
  2r-2k+1 \geq r+1$.  Then, $v_k$ and $v_{k-1}$ are only separated for
  radius $k-1$, $v_k$ and $v_{k+1}$ are only separated for radius
  $k$. So necessarily $v_{k}$ and $v_{-k}$ must be separated for
  radius $k$ or $k-1$. That means there is a vertex of the code $c'
  \not \in S$ different of $c$ at distance at most $k$ of $v_{-k}$.
  But $d(c',v_{-k})=d(c',v_{-a})+d(v_{-a},v_{-k})\geq 1+a-k \geq
  r-2k+2 \geq k+1$ (by definition of $k$), a contradiction.
    \end{proof}

\bigskip

As corollary, the code of Lemma \ref{lem:2codeex} is optimum and we have the following lower bound, 
as for light and weak codes:

\begin{corollary}\label{cor:2code}
 Let $C$ be a $(2,\lb 0,r \rb)$-code of $\cy_n$. Then $|C|\geq \left\lceil2n/s\right\rceil$ with $s=3r-\left \lfloor (r+1)/3 \right \rfloor+2$ .
\end{corollary}

It remains the case where $s$ does not divide $n$, 
with similar arguments used for light codes, one can show that:

\begin{theorem}\label{thm:2code}
Let $n,r,s,p,R$ be integers, set $k=\lfloor (r+1)/3 \rfloor$, $s=3r-k+2$ and
$n=sp+R$, with $0\leq R < s$. Then the size of an optimum $(2,\lb 0,r \rb)$-code of $\cy_n$ is:
\begin{enumerate}[noitemsep,label=\roman*)]
\item $2p$ if $R=0$;
\item $2p+1$ if $R\leq r+1$;
\item $2p+2$ otherwise.
\end{enumerate}

\end{theorem}

\section{Perspectives}

Section \ref{sec:2radius} suggests the following definition that will generalize all the previous ones:

\begin{definition}
Let $p$ be an integer and $\mathcal{R}$ be a set of non-negative integers.
A \emph{$(p,\mathcal{R})$-identifying code} of a graph $G=(V,E)$ is a subset $C$ of $\;V$ such that:
\begin{align*}
\text{(domination)}& \quad \forall x \in V, \exists r \in \mathcal{R}, B_r(x)\cap C \neq \emptyset\\
\text{(identification)} & \quad \begin{cases} 
                                   \forall x \in V, \exists R_x \subset \mathcal{R}, |R_x|\leq p, \forall y \in V, y\neq x, \exists r_{xy} \in R_x \text{ s.t.:}\\
                                   B_{r_{xy}}(x)\cap~C \neq  B_{r_{xy}}(y)\cap C
                                 \end{cases}
\end{align*}

\end{definition}

Integer $p$ corresponds to the number of radii we can assign to a vertex to separate it 
from all the 
others whereas the set $\mathcal{R}$ denotes the set of radii we can use.
This definition unifies all the previous ones: 
a $r$-identifying code is a $(1,\{r\})$-identifying code, a weak $r$-code is a
$(1,\lb 0,r \rb)$-identifying code, a light $r$-code is a $(r+1,\lb 0,r \rb)$-identifying code, a 
$r$-locating dominating code is a $(2,\{0,r\})$-identifying code.

\vspace{1em}

Proposition \ref{prop:light3} is equivalent to say that every $(p,\lb 0,r \rb)$-code in a cycle, with $p\geq 3$ 
is a $(3,\lb 0,r \rb)$-identifying code. Section \ref{sec:2radius} and Section \ref{sec:weak}
consider $(2,\lb 0,r \rb)$-identifying codes and $(1,\lb 0,r \rb)$-identifying codes of the cycle, respectively.
Hence we solved the problem of finding an optimum $(p,\lb 0,r \rb)$-identifying code (for any $p$) in a cycle.
However, the general problem of finding an optimum  $(p,\mathcal{R})$-identifying codes in the cycle is still unknown.

\bibliographystyle{plain}
\bibliography{DGMP09}

\end{document}